\documentclass[english,final,tabel]{IEEEtran}
\usepackage[T1]{fontenc}
\usepackage[latin9]{inputenc}
\usepackage{xcolor}
\usepackage{pdfcolmk}
\usepackage{bm}
\usepackage{amsmath}
\usepackage{amsthm}
\usepackage{amssymb}
\usepackage{graphicx}
\PassOptionsToPackage{normalem}{ulem}
\usepackage{ulem}

\makeatletter

\providecolor{lyxadded}{rgb}{0,0,1}
\providecolor{lyxdeleted}{rgb}{1,0,0}

\DeclareRobustCommand{\lyxdeleted}[3]{{\color{lyxdeleted}\lyxsout{#3}}}
\DeclareRobustCommand{\lyxsout}[1]{\ifx\\#1\else\sout{#1}\fi}

\theoremstyle{plain}
\newtheorem{thm}{\protect\theoremname}

\usepackage{cite}
\usepackage{amsmath,amssymb,amsfonts}
\usepackage{graphicx,psfrag,cite,subfigure}

\usepackage{algorithmic}
\usepackage{textcomp}
\usepackage{xcolor}
\def\BibTeX{{\rm B\kern-.05em{\sc i\kern-.025em b}\kern-.08em
    T\kern-.1667em\lower.7ex\hbox{E}\kern-.125emX}}
\usepackage{color} 
\usepackage{graphicx,psfrag,cite}
\usepackage{amssymb}
\usepackage{amsmath, nccmath}
\usepackage{stmaryrd}
\usepackage[acronym]{glossaries}
\newcommand{\newac}{\newacronym}
\newcommand{\ac}{\gls}

\makeglossaries
\usepackage[level=3]{wgroup_message}  

\newac{speb}{SPEB}{square position error bound}
\newac[plural=EFIMs,firstplural=Fisher information matrices (EFIMs)]{efim}{EFIM}{Fisher information matrix}
\newac{ne}{NE}{Nash equilibrium}
\newac{mse}{MSE}{mean squared error}
\newac{toa}{TOA}{time-of-arrival}
\newac{snr}{SNR}{signal-to-noise ratio}
\newac{lan}{LAN}{local area network}
\newac{psd}{PSD}{positive semidefinite}
\newac{pd}{PD}{positive definite}
\newac{wrt}{w.r.t.}{with respect to}
\newac{lhs}{L.H.S.}{left hand side}
\newac{wp1}{w.p.1}{with probability 1}
\newac{kkt}{KKT}{Karush-Kuhn-Tucker}
\newac{wlog}{w.l.o.g.}{without loss of generality}
\newac{mle}{MLE}{maximum likelihood estimation}
\newac{gps}{GPS}{global positioning system}
\newac{rssi}{RSSI}{received signal strength indicator}
\newac{mimo}{MIMO}{multiple-input multiple-output}
\newac{csi}{CSI}{channel state information}
\newac{fdd}{FDD}{frequency division duplexing}
\newac{ms}{MS}{mobile station}
\newac{bs}{BS}{base station}
\newac{d2d}{D2D}{device-to-device}
\newac{slnr}{SLNR}{signal-to-interference-leakage-and-noise-ratio}
\newac{ula}{ULA}{uniform linear antenna array}
\newac{pas}{PAS}{power angular spectrum}
\newac{mmse}{MMSE}{minimum mean square error}
\newac{zf}{ZF}{zero-forcing}
\newac{rzf}{RZF}{regularized zero-forcing}
\newac{as}{AS}{angular spread}
\newac{aod}{AOD}{angle of departure}
\newac{iid}{i.i.d.}{independent and identically distributed} 
\newac{sinr}{SINR}{signal-to-interference-and-noise ratio}
\newac{tdd}{TDD}{time-division duplex}
\newac{rvq}{RVQ}{random vector quantization}
\newac{rhs}{R.H.S.}{right hand side}
\newac{mrc}{MRC}{maximum ratio combining}
\newac{cdf}{CDF}{cumulative distribution function}
\newac{a.s.}{a.s.}{almost surely}
\newac{los}{LOS}{line-of-sight}
\newac{jsdm}{JSDM}{joint spatial division and multiplexing}
\newac{map}{MAP}{maximum a posteriori}
\newac{klt}{KLT}{Karhunen-Lo\`eve Transform}
\newac{lbe}{LBE}{link bargaining equilibrium}
\newac{se}{SE}{Stackelberg equilibrium}
\newac{uav}{UAV}{unmanned aerial vehicle}
\newac{nlos}{NLOS}{non-line-of-sight}
\newac{pdf}{PDF}{probability density function}
\newac{em}{EM}{expectation-maximization}
\newac{knn}{KNN}{$k$-nearest neighbor}
\newac{svd}{SVD}{singular value decomposition}
\newac{nmf}{NMF}{non-negative matrix factorization}
\newac{umf}{UMF}{unimodality-constrained matrix factorization}
\newac{rmse}{RMSE}{rooted mean squared error}
\newac{olos}{OLOS}{obstructed line-of-sight}
\newac{mmw}{mmW}{millimeter wave}
\newac{ber}{BER}{bit error rate}
\newac{rss}{RSS}{received signal strength}
\newac{lp}{LP}{linear program}
\newac{ufw}{U-FW}{unimodal Frank-Wolfe}
\newac{utf}{UTF}{unimodality-constrained tensor factorization}
\newac{fw}{FW}{Frank-Wolfe}
\newac{iot}{IoT}{Internet-of-Things}
\newac{mae}{MAE}{mean absolute error}
\newac{crb}{CRB}{Cram\'er-Rao bound}
\newac{aoa}{AoA}{angle of arrival}
\newac{wcl}{WCL}{weighted centroid localization}

\usepackage{color}
\usepackage{graphicx}
\usepackage{psfrag}\usepackage{cite}

\author{
Hao~Sun,~\IEEEmembership{Graduate Student Member,~IEEE},
Junting~Chen,~\IEEEmembership{Member,~IEEE}
\thanks{The work was supported in part by the NSFC under Grant No. 62171398, by the Basic Research Project No. HZQB-KCZYZ-2021067 of Hetao Shenzhen-HK S\&T Cooperation Zone, by the Shenzhen Science and Technology Program under Grant No. JCYJ20220530143804010 and No. KQTD20200909114730003, by Guangdong Research Projects No. 2019QN01X895, and by the Guangdong Provincial Key Laboratory of Future Networks of Intelligence (Grant No. 2022B1212010001).}}
\usepackage{bm}
\usepackage{bm}

\newac{svt}{SVT}{singular value thresholding}
\newac{nmse}{NMSE}{normalized mean squared error}

\renewcommand{\lyxdeleted}[3]{{\color{lyxdeleted}{}}}

\makeatother

\usepackage{babel}
\providecommand{\theoremname}{Theorem}

\begin{document}
\title{Energy-modified Leverage Sampling for Radio Map Construction via Matrix
Completion}

\maketitle
\begin{abstract}This paper explores an energy-modified leverage sampling
strategy for matrix completion in radio map construction. The main
goal is to address potential identifiability issues in matrix completion
with sparse observations by using a probabilistic sampling approach.
Although conventional leverage sampling is commonly employed for designing
sampling patterns, it often assigns high sampling probability to locations
with low \ac{rss} values, leading to a low sampling efficiency. Theoretical
analysis demonstrates that the leverage score produces pseudo images
of sources, and in the regions around the source locations, the leverage
probability is asymptotically consistent with the \ac{rss}. Based
on this finding, an energy-modified leverage probability-based sampling
strategy is investigated for efficient sampling. Numerical demonstrations
indicate that the proposed sampling strategy can decrease the \ac{nmse}
of radio map construction by more than $10$\% for both matrix completion
and interpolation-assisted matrix completion schemes, compared to
conventional methods.\end{abstract}


\section{Introduction}

Radio map construction is applied in various fields, such as wireless
network planning \cite{MoxHuaXuj:J19}, source localization \cite{CheMit:J17,GaoWuYin:J23,SunChe:C21,XinChe:C22,XinChe:J24},
UAV trajectory planning \cite{LiuChe:J23} and so on. Matrix and tensor
completions are widely used to construct radio maps from sparse and
limited samples \cite{MalZhaXuy:C18,ShrFuHong:J22,ZhaMaL:J21,SunChe:J24,ZhaFuWan:J20,SunCheLuo:C24}.
Identifiability is crucial in these methods, as the matrix or tensor
might not be completable with overly sparse measurements or improper
sampling patterns \cite{ZhaFuWan:J20,HuaLiuDuTao:J23,CheBhoSan:J15}.
Thus, optimizing the sampling pattern is essential for matrix or tensor
completion.

This paper investigates sampling strategies for constructing radio
maps through matrix completion, when some prior information is available.
Prior information can often be present in various practical situations.
In active sampling, a set of measurements is collected and analyzed
to determine future sample locations. In power-constrained sensor
networks, it may be preferable to activate only a limited number of
sensors for data reporting in each round and optimize the measurement
locations accordingly for the next rounds. Moreover, in interpolation-assisted
matrix completion \cite{SunChe:C22,CheWanZha:J23,SunChe:J22}, additional
observations can be generated from a few measurements using interpolation,
where optimizing the interpolation pattern is crucial. In these scenarios,
the initial measurements provide prior information that is essential
for optimizing the sample locations in subsequent rounds.

To optimize for the sampling pattern, Bayesian method in \cite{WanZhuLinWuq:J23}
and dictionary learning method in \cite{SheWanDinLiWu:J22} determined
the sample locations in an iterative way, but these approaches did
not utilize the online measurements to tune the sampling pattern on-the-fly.
Apart from these methods, \emph{leverage score} is commonly used to
optimize the sampling pattern \cite{CheBhoSan:J15,HuaLiuDuTao:J23,EftWakWar:J18}
for matrix completion. The leverage score is calculated from the singular
vectors of the matrix, which can be roughly estimated assuming some
prior information is available. It was theoretically shown in \cite{HuaLiuDuTao:J23}
that a biased sampling procedure that uses the leverage score to assess
the \textquotedblleft importance\textquotedblright{} of each observed
element, can recover the sparse matrix with high probability, and
there was also numerical evidence demonstrating the high efficiency
of leverage sampling compared to uniformly random sampling. The work
\cite{CheBhoSan:J15} proposed a two-phase sampling procedure for
matrices, starting with leverage score estimation and followed by
sampling for exact recovery, which requires substantially fewer samples
than uniformly random sampling method to obtain a same accuracy.

However, we discover that leverage sampling, i.e., sampling based
on the leverage probability formulated from the leverage score, is
strictly sub-optimal for radio map construction, because it may allocate
up to half of the measurements at locations where the \ac{rss} from
the sources almost vanishes. Specifically, leverage score \cite{CheBhoSan:J15,EftWakWar:J18,HuaLiuDuTao:J23}
measures the ``importance'' of sampling a row or a column of a matrix,
but \emph{not} every entry in a row or a column has the same ``importance''
to be sampled. In this paper, we construct theoretical examples to
show that the leverage score produces pseudo images of sources where
the RSS from the sources diminishes more quickly than the leverage
probability, especially when the area of the radio map scales up.
This implies that the leverage score may not be a reliable metric
for determining sampling probabilities in radio map construction.
On the other hand, in the regions around the source locations, we
show that the leverage probability is asymptotically consistent with
the \ac{rss} value of the radio map.

Based on these analysis, we propose a probabilistic sampling strategy
based on the predicted \ac{rss} values and the leverage score. The
method first constructs a rough estimate of the radio map via interpolation
or low resolution matrix completion. Then, it computes the leverage
score of each entry of the matrix to be completed. Finally, the sampling
probability is formulated based on both the RSS and the leverage score.
As a result, since leverage sampling completes the matrix with high
probability, the proposed probabilistic sampling based on energy-modified
leverage probability may also complete the matrix with high probability
and less samples. We numerically show that the proposed energy-modified
leverage sampling strategy substantially increases the accuracy of
radio map construction by over $10$\% compared to uniformly random
sampling \cite{SunChe:J22} and conventional leverage sampling methods
\cite{CheBhoSan:J15,EftWakWar:J18,HuaLiuDuTao:J23}. Integrating this
strategy with interpolation-assisted matrix completion \cite{SunChe:J22}
reduces the construction \ac{nmse} by more than $10$\% compared
to baseline methods.

\section{System Model\label{sec:System-model}}

\subsection{Propagation Model}

Consider a propagation field that is excited by $K$ sources located
at $\bm{s}_{k}\in\mathcal{D}$, $k=1,\cdots,K$, in a bounded area
$\mathcal{D}\subset\mathbb{R}^{2}$. The signal emitted from the source
is captured by $M$ sensors with known locations $\bm{z}_{m}\in\mathbb{R}^{2}$,
$m=1,2,\dots,M$, in $\mathcal{D}$. The radio map to be constructed
is modeled as
\begin{equation}
\rho(\bm{z})\triangleq\sum_{k=1}^{K}g_{k}(d(\bm{s}_{k},\bm{z}))+\zeta(\bm{z})\qquad\bm{z}\in\mathcal{D}\label{eq:model-propagation-field}
\end{equation}
where $d(\bm{s},\bm{z})=\|\bm{s}-\bm{z}\|_{2}$ describes the distance
between a source at $\bm{s}$ and a sensor at $\bm{z}$, $g_{k}(d)$
describes the propagation function from the $k$th source in terms
of the propagation distance $d$, and the term $\zeta(\bm{z})$ is
a random component that captures the spatially correlated shadowing.

The strength of the signal measured by the $m$th sensor is given
by $\gamma_{m}=\rho(\bm{z}_{m})+\epsilon_{m}$, where $\epsilon_{m}$
is a random variable with zero mean and variance $\sigma^{2}$ to
model the measurement noise.

We consider to discretize the target area $\mathcal{D}$ into $N\times N$
grid cells. Let $\bm{c}_{ij}\in\mathcal{D}$ be the center location
of the $(i,j)$th grid cell, and $\bm{H}$ be a matrix representation
of the radio map $\rho(\bm{z})$, where the $(i,j)$th entry is defined
as $H_{ij}=\rho(\bm{c}_{ij})$.

Our goal is to analyze and develop probabilistic sampling strategies
that obtain $M$ measurements $\rho(\bm{z}_{m})$ for the completion
of the matrix $\bm{H}$.

\subsection{Leverage Sampling \label{subsec:Leverage-probability}}

Given a rank-$r$ matrix $\bm{H}\in\mathbb{R}^{N\times N}$, the \ac{svd}
of $\bm{H}$ is defined as $\bm{H}=\bm{U}\bm{\Sigma}\bm{V}^{\text{T}}$.
The leverage scores $\mu_{i}$ for the $i$th row and $\nu_{j}$ for
the $j$th column are respectively defined as
\begin{equation}
\mu_{i}=N\|\bm{U}^{\text{T}}\bm{e}_{i}\|_{2}^{2}/r\label{eq:mu}
\end{equation}
\begin{equation}
\nu_{j}=N\|\bm{V}^{\text{T}}\bm{e}_{j}\|_{2}^{2}/r\label{eq:nv}
\end{equation}
where $\bm{e}_{i}$ is unit vector with $i$th element equals to $1$.

In the leverage sampling, one independently samples the grids based
on leverage probability $p_{ij}$ \cite{CheBhoSan:J15,EftWakWar:J18}
which is calculated as follows:
\begin{equation}
p_{ij}=\min\{C(\mu_{i}+\nu_{j})r\log^{2}(2N)/N,1\}\label{eq:leverage probability}
\end{equation}
where $C$ is a constant. It is shown in \cite{CheBhoSan:J15} that
under such a probabilistic sampling strategy, an arbitrary rank-$r$
matrix can be exactly recovered from $O(Nr\text{log}^{2}(N))$ observed
elements with high probability using nuclear norm minimization.

\section{Energy-modified Leverage Probability-based Probabilistic Sampling\label{sec:section two phase}}

In this section, we study a specific example to illustrate the possible
pseudo images in the conventional leverage sampling based on (\ref{eq:leverage probability}).
Then, for the informative region of the propagation field, we show
the consistency of leverage probability and the entry of $\bm{H}$.

\subsection{Existence of Pseudo Images}

For illustration purpose, we analyze the case of two sources, $K=2$,
where each source generates a rank-$1$ propagation field, denoted
as $\bm{H}^{(k)}$. Specifically, assume $g_{k}(d)$ in (\ref{eq:model-propagation-field})
takes the form of $g_{k}(d)=\alpha e^{-\beta d^{2}}$, and there is
no shadowing, i.e., $\zeta(\bm{z})=0$. To see that $\bm{H}^{(k)}$
is rank-$1$, we note from (\ref{eq:model-propagation-field}) that
\begin{align}
H_{ij}^{(k)} & =\alpha e^{-\beta((s_{x}^{(k)}-x_{i})^{2}+(s_{y}^{(k)}-y_{j})^{2})}\nonumber \\
 & =\alpha e^{-\beta(s_{x}^{(k)}-x_{i})^{2}}e^{-\beta(s_{y}^{(k)}-y_{j})^{2}}\label{eq:rank 1}
\end{align}
where $x_{i}$ and $y_{j}$ are the coordinates of the rows and the
columns, respectively. As a result, the matrix $\bm{H}^{(k)}$ can
be written as the outer product of two rank-$1$ vectors $\bm{u}_{k}=e^{-\beta((s_{x}^{(k)}-\bm{x})^{2})}$
and $\bm{v}_{k}=e^{-\beta(s_{y}^{(k)}-\bm{y})^{2}}$, scaled by $\alpha$,
where $\bm{x}$ and $\bm{y}$ are the vectors corresponding to the
coordinates of the rows and the columns, respectively.

Without loss of generality, assume the first source locates at the
origin, and the second source locates at $\bm{s}_{2}=(L_{1},L_{1})$.
We investigate the leverage sampling probability $p_{ij}$ defined
in (\ref{eq:leverage probability}) over all grid points $(i,j)$.
As shown in Fig.~ \ref{fig:The-leverage-probability}, four \textquotedblleft sources\textquotedblright{}
appears according to the values of $p_{ij}$, where two of them are
merely pseudo images. In the pseudo images, although $p_{ij}$ is
non-zero, $H_{ij}$ is essentially zero, indicating that sampling
at $(i,j)$ has negligible value for radio map construction. Pseudo
images exist because, from (\ref{eq:leverage probability}), the leverage
probability at $(i,j)$ equals to the leverage score of the $i$th
row plus the leverage score of the $j$th column of the matrix. Therefore,
in an extreme case, $K$ sources may create $K^{2}-K$ pseudo images
of the sources.

To analytically investigate the pseudo images, we define a region
$\mathcal{I}(L_{1},\delta)=\{(i,j),i\in[L_{1}+1-\delta,L_{1}+1+\delta],j\in[1,1+\delta]\}$.
The following theorem implies that $\mathcal{I}(L_{1},\delta)$ is
one of the regions of a pseudo image of the source, where the RSS
asymptotically vanishes in $\mathcal{I}(L_{1},\delta)$.
\begin{figure}
\begin{centering}
\includegraphics[width=1\columnwidth]{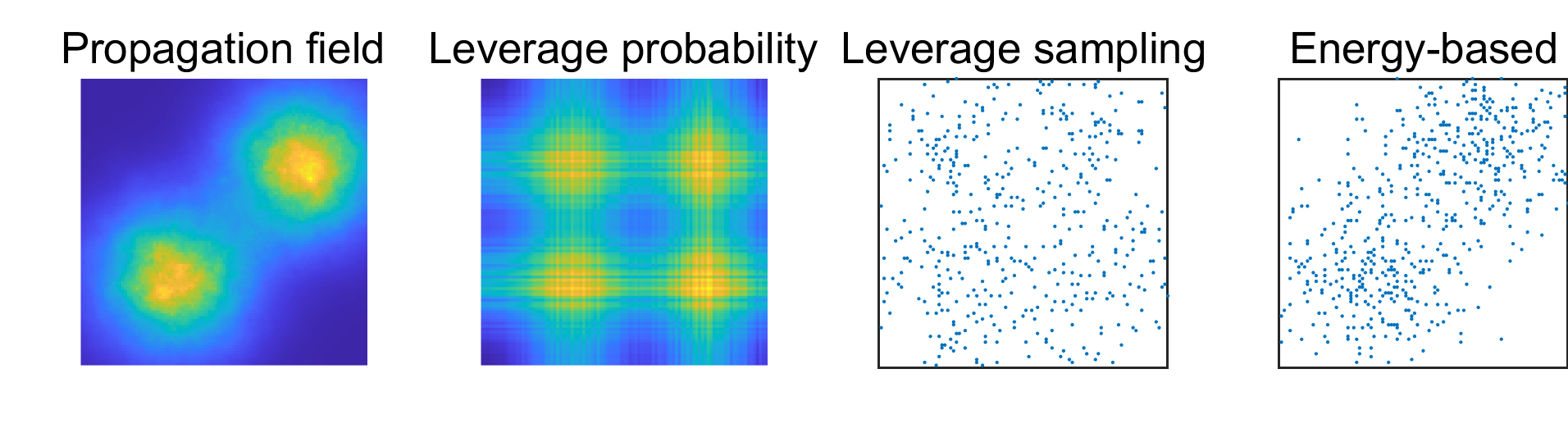}
\par\end{centering}
\caption{\label{fig:The-leverage-probability}Visible plot of propagation field,
leverage probability $p_{ij}$, leverage sampling, and energy-based
sampling.}
\end{figure}

\begin{thm}[Pseudo Image]
\label{thm:pseudo image}Consider that $\delta<L_{1}/2$. For all
$(i,j)\in\mathcal{I}(L_{1},\delta)$, $H_{ij}/p_{ij}\to0$ as $\beta\to\infty$.
\end{thm}
\begin{proof}
Denote $\bm{H}^{(k)}$ as $\bm{H}^{(k)}=\sigma_{k}\bm{u}_{k}\bm{v}_{k}^{\text{T}}$
where $\bm{u}_{k}$, $\bm{v}_{k}$ are singular vectors for the propagation
field $\bm{H}^{(k)}$ contributed from $k$th source. As a result,
$\bm{H}=\bm{H}^{(1)}+\bm{H}^{(2)}$. The \ac{svd} of $\bm{H}$ is
given by $\bm{H}=\bm{U}\bm{\Sigma}\bm{V}^{\text{T}}$, where $\bm{U}=[(\bm{u}_{1}+\bm{u}_{2})/\|\bm{u}_{1}+\bm{u}_{2}\|_{2},(\bm{u}_{1}-\bm{u}_{2})/\|\bm{u}_{1}-\bm{u}_{2}\|_{2}]$,
$\bm{V}=[(\bm{v}_{1}+\bm{v}_{2})/\|\bm{v}_{1}+\bm{v}_{2}\|_{2},(\bm{v}_{1}-\bm{v}_{2})/\|\bm{v}_{1}-\bm{v}_{2}\|_{2}]$.
From (\ref{eq:rank 1}), let $\sigma_{1}=\alpha$, there are $\bm{u}_{1}=\bm{v}_{1}=e^{-\beta\bm{x}^{2}}$
and $\bm{u}_{2}=\bm{v}_{2}=e^{-\beta(\bm{x}-L_{1})^{2}}$, for $\bm{x}=[0,1,\cdots,N-1]$.

Then, from (\ref{eq:mu}) and (\ref{eq:nv}), there are
\begin{align}
\mu_{i} & =N\left(\frac{\|\bm{u}_{1}^{\text{T}}\bm{e}_{i}+\bm{u}_{2}^{\text{T}}\bm{e}_{i}\|_{2}^{2}}{\|\bm{u}_{1}+\bm{u}_{2}\|_{2}^{2}}+\frac{\|\bm{u}_{1}^{\text{T}}\bm{e}_{i}-\bm{u}_{2}^{\text{T}}\bm{e}_{i}\|_{2}^{2}}{\|\bm{u}_{1}-\bm{u}_{2}\|_{2}^{2}}\right)\left/r\right.\nonumber \\
 & \geq N\frac{2(\|\bm{u}_{1}^{\text{T}}\bm{e}_{i}\|_{2}^{2}+\|\bm{u}_{2}^{\text{T}}\bm{e}_{i}\|_{2}^{2})}{\|\bm{u}_{1}+\bm{u}_{2}\|_{2}^{2}}/r\label{eq:mu-1}
\end{align}
 and 
\begin{align}
\nu_{j} & =N\left(\frac{\|\bm{v}_{1}^{\text{T}}\bm{e}_{j}+\bm{v}_{2}^{\text{T}}\bm{e}_{j}\|_{2}^{2}}{\|\bm{v}_{1}+\bm{v}_{2}\|_{2}^{2}}+\frac{\|\bm{v}_{1}^{\text{T}}\bm{e}_{j}-\bm{v}_{2}^{\text{T}}\bm{e}_{j}\|_{2}^{2}}{\|\bm{v}_{1}-\bm{v}_{2}\|_{2}^{2}}\right)\left/r\right.\nonumber \\
 & \geq N\frac{2(\|\bm{v}_{1}^{\text{T}}\bm{e}_{j}\|_{2}^{2}+\|\bm{v}_{2}^{\text{T}}\bm{e}_{j}\|_{2}^{2})}{\|\bm{v}_{1}+\bm{v}_{2}\|_{2}^{2}}/r.\label{eq:nv-1}
\end{align}

 Thus, for $(i,j)\in\mathcal{I}(L_{1},\delta)$, there are
\[
H_{ij}=\sum_{k=1}^{2}H_{ij}^{(k)}\leq2\alpha e^{-\beta(L_{1}-\delta)^{2}}\triangleq\bar{H}_{ij}
\]
and
\begin{align*}
p_{ij} & =\min\{C(\mu_{i}+\nu_{j})r\log^{2}(2N)/N,1\}\\
 & \geq2C\left(\frac{e^{-\beta(L_{1}-\delta)^{2}}}{\|\bm{u}_{1}+\bm{u}_{2}\|_{2}^{2}}+\frac{e^{-\beta\delta^{2}}}{\|\bm{v}_{1}+\bm{v}_{2}\|_{2}^{2}}\right)\log^{2}(2N)\\
 & \triangleq\bar{p}_{ij}.
\end{align*}

As a consequence, there are
\begin{align*}
 & \frac{H_{ij}}{p_{ij}}\leq\frac{\bar{H}_{ij}}{\bar{p}_{ij}}=\frac{2\alpha e^{-\beta(L_{1}-\delta)^{2}}}{2C\left(\frac{e^{-\beta(L_{1}-\delta)^{2}}}{\|\bm{u}_{1}+\bm{u}_{2}\|_{2}^{2}}+\frac{e^{-\beta\delta^{2}}}{\|\bm{v}_{1}+\bm{v}_{2}\|_{2}^{2}}\right)\log^{2}(2N)}\\
= & \frac{\alpha}{C\log^{2}(2N)}\left(\frac{1}{\frac{1}{\|\bm{u}_{1}+\bm{u}_{2}\|_{2}^{2}}+\frac{e^{-\beta\delta^{2}}}{\|\bm{v}_{1}+\bm{v}_{2}\|_{2}^{2}e^{-\beta(L_{1}-\delta)^{2}}}}\right)\to0
\end{align*}
as $\beta$$\to$$\infty$.
\end{proof}
According to Theorem \ref{thm:pseudo image}, when $\beta\to\infty$,
corresponding to a sharp propagation field, both $p_{ij}$ and $H_{ij}$
tend to $0$, but $H_{ij}$ tends to $0$ more rapidly than $p_{ij}$.
This implies that if one uses the leverage probability $p_{ij}$ in
(\ref{eq:leverage probability}) as the sampling probability, the
probability $p_{ij}$ in $\mathcal{I}(L_{1},\delta)$ may still be
non-negligible but the measurement $H_{ij}$, which is essentially
$0$, contains almost no information of the propagation field. As
a result, sampling at the pseudo image is highly inefficient.

While Theorem \ref{thm:pseudo image} examines the case of increasing
$\beta$ to a large value resulting in a sharp propagation field,
this situation is analogous to increasing the distance $L_{1}$ while
keeping $\beta$ constant. Likewise, pseudo images appear where the
\ac{rss} of the sources vanishes, and sampling at these pseudo images
is inefficient.

\subsection{Consistency of the Leverage Probability $p_{ij}$ and $H_{ij}$}

For the regions around the sources, the leverage probability $p_{ij}$
defined in (\ref{eq:leverage probability}) is essentially consistent
with $H_{ij}$.

Define $\mathcal{J}(L_{1},\delta)=\{(i,j),i,j\in[L_{1}+1-\delta,L_{1}+1+\delta]\}$.
The following theorem implies that $\mathcal{J}(L_{1},\delta)$ is
one of the regions of high importance, where the leverage probability
is strongly correlated with $\bm{H}$ in $\mathcal{J}(L_{1},\delta)$.
\begin{thm}[Consistency]
\label{thm:thm 2}For $(i,j)\in\mathcal{J}(L_{1},\delta)$, $H_{ij}/p_{ij}\to C^{\prime}/C$
with $C^{\prime}=\alpha/(4\text{\emph{log}}^{2}(2N))$, as $\delta\to0$.
\end{thm}
\begin{proof}
From (\ref{eq:rank 1}), there are
\begin{equation}
H_{ij}=\alpha e^{-\beta(i-1)^{2}}e^{-\beta(j-1)^{2}}+\alpha e^{-\beta(i-L_{1}-1)^{2}}e^{-\beta(j-L_{1}-1)^{2}}.\label{eq:Hij}
\end{equation}
For the leverage probability $p_{ij}$ defined in (\ref{eq:leverage probability}),
from (\ref{eq:mu-1}) and (\ref{eq:nv-1}), there are
\begin{align}
p_{ij} & =2C\text{log}^{2}(2N)(e^{-\beta2(i-1)^{2}}+e^{-\beta2(i-L_{1}-1)^{2}}\nonumber \\
 & \qquad\qquad\qquad+e^{-\beta2(j-1)^{2}}+e^{-\beta2(j-L_{1}-1)^{2}}).\label{eq:p_ij}
\end{align}
Then, for $(i,j)\in\mathcal{J}(L_{1},\delta)$, and $\delta\to0$,
there are $e^{-\beta2(i-1)^{2}}/e^{-\beta2(j-1)^{2}}\to1$ and $(e^{-\beta2(i-1)^{2}}+e^{-\beta2(j-1)^{2}})/2e^{-\beta(i-1)^{2}}e^{-\beta(j-1)^{2}}\to1$.
Thus, there are $H_{ij}/p_{ij}\to C'/C,$ where $C^{\prime}=$$\alpha/(4\text{log}^{2}(2N))$.
\end{proof}
Thus, for the region of high importance, the $H_{ij}$ itself is essentially
consistent as the leverage probability $p_{ij}$.

The results in Theorems \ref{thm:pseudo image} and \ref{thm:thm 2}
motivate the proposed probabilistic sampling based on the energy-modified
leverage probability

\begin{equation}
\tilde{p}_{ij}=C_{1}\sqrt{H_{ij}p_{ij}}\label{eq:energy modified}
\end{equation}
where $p_{ij}$ is from (\ref{eq:leverage probability}) and $C_{1}$
is a constant that depends on the matrix dimension and the rank structure.

It follows that at the region of high importance $\mathcal{J}$, there
is $\tilde{p}_{ij}\propto p_{ij}$, according to Theorem \ref{thm:thm 2}.
Therefore, $\tilde{p}_{ij}$ is essentially the leverage probability
in $\mathcal{J}$. In the pseudo images $\mathcal{I}$ where the source
signal almost vanishes, there is $\tilde{p}_{ij}\ll p_{ij}$, thus,
the sampling probability is significantly reduced, and therefore,
frequently sample at locations where the source signal vanishes can
be avoided.
\begin{figure*}[t]
\subfigure[]{\includegraphics[width=0.33\textwidth]{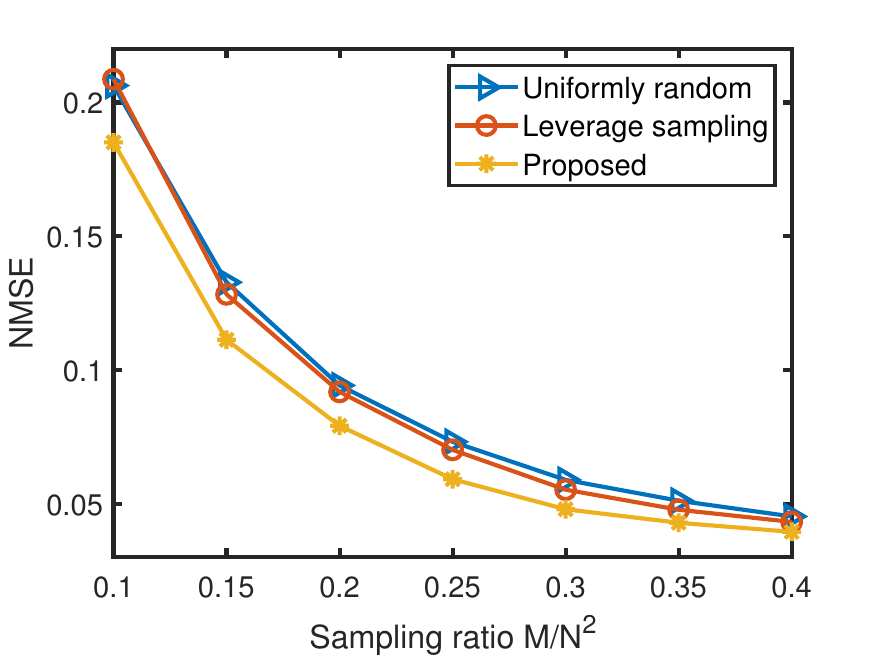}}\subfigure[]{\includegraphics[width=0.33\textwidth]{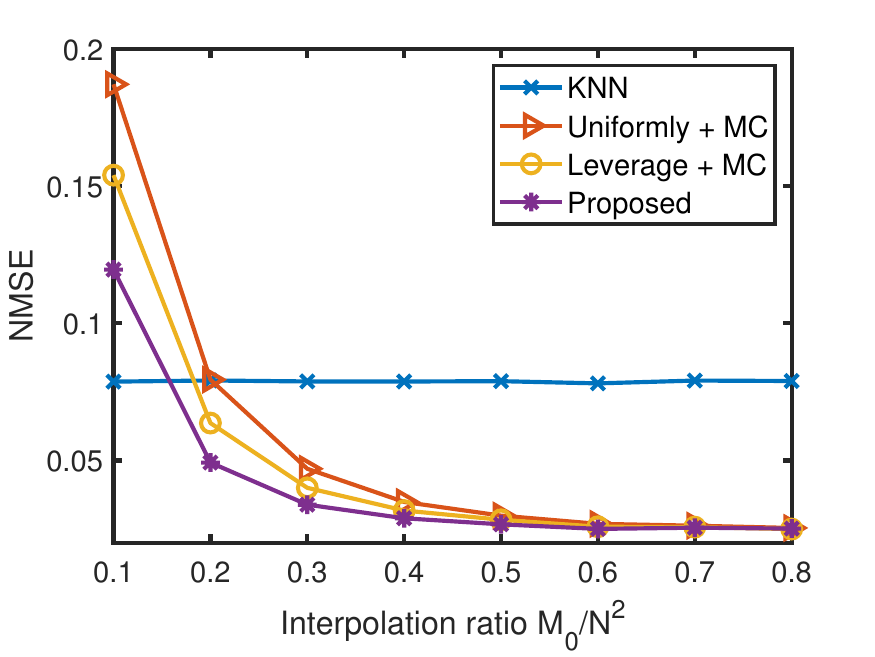}}\subfigure[]{\includegraphics[width=0.33\textwidth]{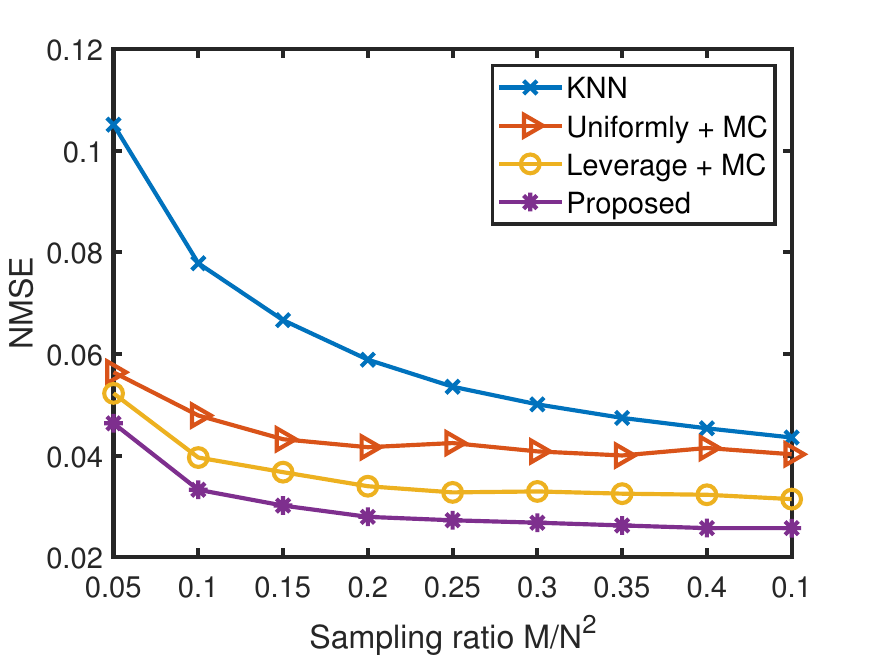}}\caption{\label{fig: matrix compleiton}(a) Construction \ac{nmse} of matrix
completion under different sampling strategy versus different sampling
ratio $M/N^{2}$. (b) Construction \ac{nmse} of interpolation-assisted
matrix completion under different interpolation radio $M_{0}/N^{2}$.
(c) Construction \ac{nmse} of interpolation-assisted matrix completion
under different sampling ratio $M/N^{2}$.}
\end{figure*}

\subsection{Implementation of Energy-modified Leverage Sampling\label{subsec:Implementation-of-Energy-modifie}}

Suppose that the expected number of measurements is $M$. First, we
interpolate a matrix $\hat{\bm{H}}$, using $\iota M$ measurements
uniformly random taken in the area of interest, for a small $\iota<1$,
as the prior information. Algorithms such as \ac{knn}, Kriging, or
regression can be used for the construction of $\hat{\bm{H}}$. Then,
we calculate the \ac{svd} of $\hat{\bm{H}}=\hat{\bm{U}}\hat{\bm{\Sigma}}\hat{\bm{V}}^{\text{T}}$
to obtain the leverage scores, $\hat{\mu}_{i}$ and $\hat{\nu}_{j}$
as in (\ref{eq:mu})\textendash (\ref{eq:nv}), for the $i$th row
and $j$th column, and we further obtain $\hat{p}_{ij}=\hat{\mu}_{i}+\hat{\nu}_{j}$.

Next, we establish the energy-modified leverage probability $\tilde{p}_{ij}=C_{1}\sqrt{\hat{H}_{ij}\hat{p}_{ij}}$,
where $C_{1}=(1-\iota)M(\sum_{i,j}\sqrt{\hat{H}_{ij}\hat{p}_{ij}})$.
Then, in the second round of sampling, we independently sample each
grid according to the probability $\tilde{p}_{ij}$. One can verify
that the expected number of samples equals to $(1-\iota)M$. Finally,
matrix completion is performed via nuclear norm minimization using
the total $M$ measurements obtained from the two rounds of sampling,
to obtain the reconstructed ratio map $\tilde{\bm{H}}$.

\section{Numerical Results\label{sec:Numerical-Results}}

We adopt model (\ref{eq:model-propagation-field}) to simulate the
radio map in an $L\times L$ area for $L=2$ kilometers, $K=3$ sources,
where $g_{k}(d)=Pd^{-1.5}A(f)^{-d}$, with parameter $A(f)=0.8$,
corresponding to an empirical energy field of underwater acoustic
signal at frequency $f=5$ kHz, $d=\sqrt{x^{2}+y^{2}+h^{2}}$ represents
the distance from the source, $(x,y)$ is the coordinate, and $h=400$
meters is the depth of interest \cite{BRELeoMak:B04}. The shadowing
component in log-scale $\mbox{log}_{10}\zeta$ is modeled using a
Gaussian process with zero mean and auto-correlation function $\mathbb{E}\{\mbox{log}_{10}\zeta(\bm{z}_{i})\mbox{log}_{10}\zeta(\bm{z}_{j})\}=\sigma_{\text{s}}^{2}\mbox{exp}(-||\bm{z}_{i}-\bm{z}_{j}||_{2}/d_{\text{c}})$,
in which $d_{\text{c}}=200$ meters, $\sigma_{s}^{2}=1$. We choose
$\epsilon\sim\mathcal{N}(0,\sigma)$ with $\sigma=0.1$ to model the
measurement noise.

The \ac{nmse} of the reconstructed radio map is employed for performance
evaluation, which is calculated through $||\tilde{\bm{H}}-\bm{H}||_{F}^{2}/||\bm{H}||_{F}^{2}$.
We evaluate the performance of radio map construction under matrix
dimension $N=100$.

We first test the energy-modified leverage sampling under the matrix
completion scheme in \cite{CaiCan:J10}, and choose $M$ measurements
with the sampling ratio $M/N^{2}=10\text{\%}-40\text{\%}$. We compare
three sampling schemes, the proposed energy-modified leverage sampling,
the uniformly random sampling, and the leverage sampling. For the
leverage sampling, we first uniformly random sample $\iota M$ grids
to obtain $p_{ij}$, then we sample $(1-\iota)M$ grids accordingly,
where we choose $\iota=0.7$. Fig.~\ref{fig: matrix compleiton}
(a) shows that the proposed energy-modified leverage sampling outperforms
the leverage sampling and uniformly random sampling in construction
NMSE with larger than $10$\% improvement under small sampling ratio.

We then test the effectiveness of the proposed energy-modified leverage
sampling in interpolation-assisted matrix completion \cite{SunChe:J22}.
In this method, we choose $\iota=1$ to estimate $\hat{H}_{ij}$ and
$\hat{p}_{ij}$, based on all the measurements $M$, then, we interpolate
$M_{0}$ grids based on the probability $\tilde{p}_{ij}$ with $C_{1}=M_{0}(\sum_{i,j}\sqrt{\hat{H}_{ij}\hat{p}_{ij}})$.
After that, we perform matrix completion. We compare this method to
the following baseline methods. Baseline 1: Uniformly random interpolation
followed by matrix completion (Uniformly + MC). In this approach,
we uniformly and randomly interpolate grids based on the $M$ measurements,
and then, use the \ac{svt} algorithm to solve the matrix completion
problem. Baseline 2: Leverage interpolation followed by matrix completion
(Leverage + MC). This approach interpolates each grids independently,
based on the leverage probability $p_{ij}$. Baseline 3: \ac{knn}
method, with $k=3$.

To show the influence of interpolation ratio $M_{0}/N^{2}$ on the
performance of the proposed method, we choose the number of measurements
$M=1000$ to satisfy the sampling ratio $M/N^{2}=10\text{\%}$ and
vary the interpolation ratio $M_{0}/N^{2}=10\text{\%}-80\text{\%}$.
Simulation results in Fig.~\ref{fig: matrix compleiton} (b) demonstrates
that under the same sampling ratio, the proposed energy-modified leverage
sampling significantly outperforms the baseline methods by more than
$20$\% under a low interpolation ratio. The proposed method only
needs a interpolation ratio of among $50$\% to attain the best performance.

To show the influence of the sampling ratio $M/N^{2}$ on the performance
of the proposed method, we choose the number of measurements $M$
to satisfy the sampling ratio $M/N^{2}$ ranging from $5\text{\%}-45\text{\%}$
and set the number of grids to be interpolated $M_{0}=3000$ with
$M_{0}/N^{2}=30\text{\%}$. Fig.~\ref{fig: matrix compleiton} (c)
illustrates that interpolation based on energy-modified leverage probability
exhibits an over $10$\% improvement of construction accuracy, as
compared to the baseline methods. It demonstrates superior performance
compared to the conventional matrix completion approach, which lacks
a designed interpolation pattern, effectively highlighting the effectiveness
of the proposed sampling strategy.

\section{Conclusion\label{sec:Conclusion}}

This paper proposed an energy-modified leverage sampling method. It
was theoretically shown that the leverage scores were not efficient
since the existence of pseudo images. In addition, the leverage probability
$p_{ij}$ was shown to be consistent with the RSS at the regions around
source locations. Then, an energy-modified leverage probability $\tilde{p}_{ij}$
was formulated. Simulation results demonstrated the proposed method
have over $10$\% improvement in construction \ac{nmse}. \vfill{}
\bibliographystyle{IEEEtran}
\bibliography{IEEEabrv,StringDefinitions,JCgroup,ChenBibCV}

\end{document}